\documentclass[11pt]{article}
\usepackage{graphicx}
\usepackage{verbatim}
\usepackage{fullpage}
\usepackage{hyperref}
\usepackage{amssymb}
\usepackage{amsmath}
\usepackage{amsthm}
\newtheorem{theorem}{Theorem}
\newtheorem{prop}[theorem]{Proposition}
\newtheorem{lemma}[theorem]{Lemma}

\def\reals{\mathbb{R}}
\def\R{\reals}
\def\sph{\mathbb{S}}
\def\suchthat{\;:\;}

\def\eps{\epsilon}
\def\diam{\mathrm{D}}

\newcommand{\ball}[2]{#2\mathbb{B}_{#1}}
\newcommand{\abs}[1]{\left|#1\right|}
\newcommand{\norm}[1]{\left\|#1\right\|}
\newcommand{\sqnorm}[1]{\left\|#1\right\|^{2}}

\newcommand{\prob}[1]{{\sf Pr}\left(#1\right)}

\newcommand{\vol}[1]{\operatorname{vol}\left(#1\right)}

\title{Sampling Harmonic Concave Functions: The Limit of Convexity Based Isoperimetry}
\author{Karthekeyan Chandrasekaran \thanks{School of Computer Science, Georgia Tech. Email: {\tt karthe@gatech.edu, vempala@cc.gatech.edu}}
\and
Amit Deshpande\thanks{Microsoft Research. Email: {\tt amitdesh@microsoft.com}}
\and
Santosh Vempala\footnotemark[1]
}
\date{}

\begin{document}
\maketitle

\begin{abstract}
Logconcave functions represent the current frontier of efficient algorithms for sampling, optimization and integration in $\reals^n$ \cite{LV06}.
Efficient sampling algorithms to sample according to a probability density (to which the other two problems can be reduced) relies on good isoperimetry which is known to hold for arbitrary logconcave densities. In this paper, we extend this frontier in two ways: first, we characterize convexity-like conditions that imply good isoperimetry, i.e., what condition on function values along every line guarantees good isoperimetry? The answer turns out to be the set of $(1/(n-1))$\emph{-harmonic concave functions} in $\reals^n$; we also prove that this is the best possible characterization along every line, of functions having good isoperimetry. Next, we give the first efficient algorithm for sampling according to such functions with complexity depending on a smoothness parameter. Further, noting that the multivariate Cauchy density is an important distribution in this class, we exploit certain properties of the Cauchy density to give an efficient sampling algorithm based on random walks with a mixing time that matches the current best bounds known for sampling logconcave functions.
\end{abstract}


\section{Introduction}

Over the past two decades, logconcave functions have emerged as the common frontier for the complexity of sampling, optimization and integration. More precisely, given a function $f:\R^n \rightarrow \R_+$, accessible by querying the function value at any point $x \in \R^n$, and an error parameter $\eps > 0$, three fundamental problems are: (i) Integration: estimate $\int f$ to within $1\pm \eps$, (ii) Maximization: find $x$ that approximately maximizes $f$, i.e., $f(x) \ge (1-\eps)\max f$, and (iii) Sampling: generate $x$ from density $\pi$ with $d_{tv}(\pi,\pi_f) \le \eps$ where $d_{tv}$ is the total variation distance and $\pi_f$ is the density proportional to $f$.
(For each of these, exact solutions are intractable.)

The complexity of an algorithm is measured by the number of queries for the function values and the number of arithmetic operations. The most general class of functions for which these problems are known to have polynomial complexity in the dimension $n$, is the class of logconcave functions. A function $f:\R^n \rightarrow \R_+$ is logconcave if its logarithm is concave on its support, i.e., for any two points $x,y \in \R^n$ and any $\lambda \in (0,1)$,
\begin{equation}\label{logconcave}
f(\lambda x + (1-\lambda) y) \ge f(x)^\lambda f(y)^{1-\lambda}.
\end{equation}
This powerful class generalizes indicator functions of convex bodies (and hence the problems subsume convex optimization and volume computation) as well as Gaussians.
Following the celebrated result of Dyer, Frieze and Kannan \cite{DFK} giving a polynomial algorithm for estimating the volume of a convex body, a long line of work \cite{AK,L90,DF90,LS92,LS93,KLS97,LV1,LV2,LV3} culminated in the results that both sampling and integration have polynomial complexity for any logconcave density. Integration is done by a reduction to sampling and sampling also provides an alternative to the Ellipsoid method for optimization \cite{BV,KV,LV06}. Sampling itself is achieved by a random walk whose stationary distribution has density proportional to the given function. The key question is thus the rate of convergence of the walk, which depends (among other things) on the isoperimetry of the target function. Roughly speaking, isoperimetry is the minimum ratio of the measure of the boundary of a partition of space into two sets to the measure of the smaller of the two sets. Logconcave functions satisfy the following isoperimetric inequality:
\begin{theorem} \cite{DF90,LS93} \label{thm:lcIsoperimetry}
Let $\pi_f$ be a distribution in $\R^n$ with density proportional to a logconcave function $f$. Let $K$ be the support of $f$, $D$ the diameter of $K$ and $S_1,S_2,S_3$ be any partition of $K$. Then,
\[
\pi_f(S_3) \ge \frac{2d(S_1,S_2)}{D}\min \{ \pi_f(S_1), \pi_f(S_2)\}.
\]
where $d(S_1,S_2)$ refers to the minimum distance between any two points in $S_1$ and $S_2$.
\end{theorem}

While these results are fairly general, they do not capture the complete class of functions which have good isoperimetry. Logconcave functions in $\reals^n$ are defined by a convex-combination based condition for every two points in the support of the function, saying that the function is logconcave along every line.  
This is a generalisation of the case of convex bodies where we have that the line segment connecting any two points in the body lies completely within the body. This motivates the following question: What is the condition that needs to be satisfied along every line by a function, for it to have good isoperimetry?

In this paper, we present the complete class of functions with good isoperimetry that can be described by such convex-combination based conditions. We also give an efficient algorithm to sample from these functions. Further, we identify the Cauchy density (which is not logconcave) as a well-known example in this set of functions and obtain an efficient algorithm for sampling the multivariate Cauchy density restricted to a convex body. Its complexity matches the best-known bounds for logconcave functions. We note that the density functions satisfying the convex-combination based characterization that we present here, could be heavy-tailed with unbounded moments as is the case with the Cauchy density.


To motivate and state our results, we begin with a discussion of $1$-dimensional conditions.

\subsection{From concave to quasi-concave}

A function $f: \reals^{n} \rightarrow \reals_+$ is said to be
\[
\begin{cases}
\text{\emph{concave} if},
& f\left(\lambda x +(1-\lambda)y\right) \geq \lambda f(x) + (1-\lambda)f(y)\\
& \\

\text{\emph{logconcave} if},
& f\left(\lambda x +(1-\lambda)y\right) \geq f(x)^\lambda f(y)^{1-\lambda}\\
&\\

\text{\emph{$s$-harmonic-concave} if},
& f\left(\lambda x +(1-\lambda)y\right) \geq \left({\frac{\lambda}{f(x)^s}+\frac{1-\lambda}{f(y)^s}}\right)^{-\frac{1}{s}}\\
&\\

\text{\emph{quasi-concave} if},
& f\left(\lambda x +(1-\lambda)y\right) \geq \min\{f(x), f(y)\}\\
\end{cases}\\
\]
for all $\lambda \in [0,1], \forall x,y \in \R^n$.\\

These conditions are progressively weaker, restricting the function value at a convex combination of $x$ and $y$ to be at least the arithmetic average, geometric average, harmonic average and minimum, respectively. Note that $s=1$ gives the usual harmonic average (and $s_1$-harmonic-concave functions are also $s_2$-harmonic-concave if $s_1 < s_2$). It is thus easy to verify that:
\[
\text{concave} \subsetneq \text{logconcave} \subsetneq \text{$s$-harmonic-concave} \subsetneq \text{quasi-concave}.
\]
Relaxing further would violate unimodality, i.e., there could be two distinct local maxima, which appears quite problematic for all of the fundamental problems. Also, it is well-known that quasi-concave functions have poor isoperimetry.


We note here the relation between the $s$-harmonic-concave probability density function and the probability measure as shown by C. Borell \cite{B74,B75}. This gives an equivalence between the one-dimensional convexity-like condition to a condition on the probability measure of the function.
\begin{lemma}\label{lemma:borell}
Let $-\infty< \kappa \leq \frac{1}{n}$. An absolutely continuous probability measure $\mu$ on $\reals ^n$ is $\kappa$-concave if and only if it is concentrated on an open convex set $K$ in $\reals ^n$ and has there a positive density $p$, which is $\kappa(n)-concave$ for $\kappa(n)=\frac{\kappa}{1-\kappa n}$.
\end{lemma}
Thus, we have that if the density function is $s$-harmonic-concave for $s\in[0,1/n]$, then the corresponding probability measure is $\kappa$-concave for $\kappa=\frac{-s}{1-ns}$. Further, Bobkov \cite{bobkov} proves the following isoperimetric inequality for $\kappa$-concave probability measures for $-\infty<\kappa\leq 1$.
\begin{theorem} Given a $\kappa$-concave probability measure $\mu$, for any measurable subset $A\subseteq\reals^n$,
\[
\mu(\delta A)\geq \frac{c(\kappa)}{m}\min\{\mu(A),1-\mu(A)\}^{1-\kappa}
\]
where $m=\int_{\reals^n}|x|d\mu(x)$, for some constant $c(\kappa)$ depending on $\kappa$.
\end{theorem}
Using the characterisation given by C.Borell in Lemma \ref{lemma:borell}, one can obtain a weaker form of isoperimetric inequality from the above theorem for $s$-harmonic-concave functions $f:\reals^n\rightarrow \reals_+$ for $s\geq \frac{1}{n}$ to say that for any measurable subset $A$, for some constant $c(s)$ depending only on $s$, 
\[
\pi_f(\delta A)\geq \frac{c(s)}{m}\min\{\pi_f(A),1-\pi_f(A)\}^{1+\frac{s}{1-ns}}
\]
We note that our result gives a stronger inequality, in the sense that, we remove the dependence on $s$ from the inequality completely. We prove such an inequality for the more general class of $\left(\frac{1}{n-1}\right)$-harmonic-concave functions rather than $\left(\frac{1}{n}\right)$-harmonic-concave functions. We proceed to show that the limit of isoperimetry for convex-combination based conditions along every line for the density function, is the set of $(1/(n-1))$-harmonic-concave functions. Further, we address the problem of sampling from such  distributions restricted to any convex set (possibly unbounded).

\subsection{The Cauchy density} \label{sec:cauchy}
The generalized Cauchy probability density $f: \reals^{n} \rightarrow \reals_{+}$ can be written as
\[
f(x) \propto \frac{\det(A)^{-1}}{\left(1 + \sqnorm{A(x-m)}\right)^{(n+1)/2}}.
\]
where $A\in \reals_{n\times n}$. For simplicity, we assume $m = \bar{0}$ using a translation.
It is easy to sample this distribution in full space (by an affine transformation it becomes spherically symmetric and therefore a one-dimensional problem) \cite{johnson87Book}. We consider the problem of sampling according to the Cauchy density restricted to a convex set. This is reminiscent of the work of Kannan and Li who considered the problem of sampling a Gaussian distribution restricted to a convex set \cite{kannanLi}.

The Cauchy density function belongs to the broader class of L\'{e}vy skew alpha-stable distributions (or simply known as stable distributions) \cite{mandel,nolan} which are useful in modeling many variables in physics and mathematical finance. Unlike most stable distributions, Cauchy densities have a closed form expression for their density.
Being a $1$-stable distribution, it finds useful applications in a variety of problems including the approximate nearest neighbor problem and dimension reduction in $l_{1}$ norm \cite{indyk06,diim,li07}.

\subsection{Our results}

Our first result establishes good isoperimetry for $1/(n-1)$-harmonic-concave functions in $\R^n$.

\begin{theorem} \label{thm:isoperimetric}
Let $f: \reals^{n} \rightarrow \reals_{+}$ be a $(1/(n-1))$-harmonic-concave function with a support $K$ . Let $\reals ^n = S_{1} \cup S_{2} \cup S_{3}$ be a measurable partition of $\reals ^n$ into three non-empty subsets. Then
\[
\pi_{f}(S_{3}) \geq \frac{d(S_{1}, S_{2})}{D} \min\left\{\pi_{f}(S_{1}), \pi_{f}(S_{2})\right\},
\]
where $D$ is the diameter of $K$.
\end{theorem}
It is worth noting that the isoperimetric coefficient above is only smaller by a factor of $2$ when compared to that of logconcave functions (Theorem \ref{thm:lcIsoperimetry}).

Next, we prove that if we go slightly beyond the class of $(1/(n-1))$-harmonic-concave functions, then there exist functions with exponentially small isoperimetric coefficient.

\begin{theorem}\label{thm:limitIsoperimetry}
For any $\eps > 0$, there exists a $1/(n-1-\eps)$-harmonic-concave function $f: \reals^{n} \rightarrow \reals_{+}$ with a convex support $K$ of finite diameter (i.e., $D_{K} < \infty$) and a partition $\reals^{n} = S \cup T$ such that
\[
\frac{\pi_{f}(\partial S)}{\min\left\{\pi_{f}(S), \pi_{f}(T)\right\}} \leq Cn(1+\eps)^{-\eps n}
\]
for some constant $C>0$.
\end{theorem}

To summarize, we have the following table for isoperimetry:\\
\begin{center}
\begin{tabular}{|c|c|}
\hline
Nature of f & Good Isoperimetry? \\
\hline
Concave & Yes \\
Logconcave & Yes \\
$(1/(n-1))$-harmonic-concave & Yes \\
$(1/(n-1-\eps))$-harmonic-concave ($\eps > 0 $) & No \\
Harmonic-Concave & No \\
Quasi-Concave & No \\
\hline
\end{tabular}
\end{center}

Next we prove that the ball walk with a Metropolis filter can be used to sample efficiently according to the $(1/(n-1))$-harmonic concave distribution function which satisfy a certain Lipschitz condition. At a point $x$, the ball walk picks a new point $y$ uniformly at random from a fixed radius ball around $x$ and moves to $y$ with probability $\min\{1,f(y)/f(x)\}$. A distribution $\sigma_0$ is said to be an $H$-warm start ($H>0$) for the distribution $\pi_f$ if for all $S\subseteq \reals^n$, $\sigma_0(S)\leq H\pi_f(S)$. Let $\sigma_m$ be the distribution after $m$ steps of the ball walk with a Metropolis filter.

We say that a function $f:\reals^n\rightarrow \reals_+$ has parameters $(\alpha,\delta)$ if for all points $x,y$ in the support of $f$ such that $\norm{x-y}\leq \delta$, we have $\max\{f(x)/f(y),f(y)/f(x)\}\leq \alpha$.
%
\begin{theorem}\label{thm:sampling_hc}
Let $f:\reals^n\rightarrow \reals_+$ be proportional to an $s$-harmonic-concave function with parameters $(\alpha,\delta)$, restricted to a convex body $K\subseteq \reals^n$ of diameter $D$, where $s\leq 1/(n-1)$. Let $K$ contain a ball of radius $\delta$ and $\sigma_0$ be an $H$-warm start. Then, there exists a radius $r$ for the ball walk such that, after
\[
m\geq \left(\frac{CnD^2}{\delta^2}\log{\frac{2H}{\eps}}\right)\cdot \max\left\{\frac{nH^2}{\eps^2},\frac{(\alpha^{s}-1)^2}{s^2\delta^2}\right\}
\]
steps, we have that
\[
d_{tv}(\sigma_m,\pi_f)\leq \eps,
\]
for some absolute constant C, where $d_{tv}(\cdot,\cdot)$  is the total variation distance.
\end{theorem}

Applying the above theorem directly to sample according to the Cauchy density, we get a  mixing time of $O\left(\left(\frac{n^3H^2}{\eps^2}\log{\frac{2H}{\eps}}\right)\cdot \max\left\{\frac{H^2}{\eps^2},n\right\}\right)$ using parameters $D=\frac{8\sqrt{2n}H}{\eps}$ (since the probability measure outside the $D$-ball is at most $\eps/2H$), $\delta=1$ and $\alpha=e^{\frac{n+1}{2}}$. Using a more careful analysis (comparison of $1$-step distributions), this bound can be improved to match the current best bounds for sampling logconcave functions.

\begin{theorem}\label{thm:cauchy_sampling}
Let $f$ be proportional to a Cauchy probability density restricted to a convex set $K\subseteq \reals ^n$ containing a ball of radius $\|A^{-1}\|_2$ and let $\sigma_0$ be an $H$-warm starting distribution. Then after
\[
m \geq O\left(\frac{n^3H^4}{\epsilon ^4}\log\frac{2H}{\epsilon}\right)
\]
steps with ball-walk radius $r=\eps/8\sqrt{n}$, we have
\[
d_{tv}(\sigma_{m},\pi_{f}) \leq \epsilon
\]
where $d_{tv}(.,.)$ is the total variation distance.
\end{theorem}

The proof of this theorem departs from earlier counterparts in a significant way. In addition to isoperimetry, and the closeness of one-step distributions of nearby points, we have to prove that most of the measure is contained in a ball of not-too-large radius. For logconcave densities, this large-ball probability decays exponentially with the radius. For the Cauchy density it only decays linearly (Proposition \ref{prop:large-ball}).

\section{Preliminaries}

Let $\ball{x}{r}$ denote a ball of radius $r$ around point $x$.
One step of the ball walk at a point $x$ defines a probability distribution $P_{x}$ over $\reals^{n}$ as follows.
\[
P_{x}(S) = \int_{S \cap \ball{x}{r}} \min\left\{1, \frac{f(y)}{f(x)}\right\} dy.
\]
For every measurable set $S \subseteq \reals^{n}$ the ergodic flow from $S$ is defined as
\[
\Phi(S) = \int _{S} P_{x}(\reals^{n} \setminus S) {f}(x) dx,
\]
and the measure of $S$ according to $\pi_f$ is defined as $\pi _f(S) = \int _{S} f(x)dx/\int _{\reals ^n}f(x)dx$. The $s$-conductance $\phi_s$ of the Markov chain defined by ball walk is
\[
\phi_{s} = \inf_{s \leq \pi_{f}(S) \leq 1/2} \frac{\Phi(S)}{\pi_{f}(S) - s}.
\]
To compare two distributions $Q_1$, $Q_2$ we use the total variation distance between $Q_1$ and $Q_2$, defined by $d_{tv}(Q_1,Q_2)=\sup_{A} |Q_1(A)-Q_2(A)|$. When we refer to the distance between two sets, we mean the minimum distance between any two points in the two sets. That is, for any two subsets $S_1, S_2\subseteq \reals^n$, $d(S_1,S_2):=\min\{|u-v|:u\in S_1, v\in S_2\}$.
Next we quote a lemma from \cite{LS93} which relates the $s$-conductance to the mixing time.
\begin{lemma} \label{lemma:mix}
Let $0 < s \leq 1/2$ and $H_{s} = \sup_{\pi_{f}(S) \leq s} \abs{\sigma_{0}(S) - \pi_{f}(S)}$. Then for every measurable $S \subseteq \reals^{n}$ and every $m \geq 0$,
\[
\abs{\sigma_{m}(S) - \pi_{f}(S)} \leq H_{s} + \frac{H_{s}}{s} \left(1 - \frac{\phi_{s}^{2}}{2}\right)^{m}.
\]
\end{lemma}
Finally, the following localization lemma \cite{LS93,KLS95} is a useful tool in the proofs of isoperimetric inequalities.
\begin{lemma} \label{lemma:local}
Let $g: \reals^{n} \rightarrow \reals$ and $h: \reals^{n} \rightarrow \reals$ be two lower semi-continuous integrable functions such that
\[
\int_{\reals^{n}} g(x) dx > 0 \quad \text{and} \quad \int_{\reals^{n}} h(x) dx > 0.
\]
Then there exist two points $a, b \in \reals^{n}$ and a linear function $l: [0, 1] \rightarrow \reals_{+}$ such that
\[
\int_{0}^{1} g((1-t)a + tb) l(t)^{n-1} dt > 0 \quad \text{and} \quad \int_{0}^{1} h((1-t)a + tb) l(t)^{n-1} dt > 0.
\]
\end{lemma}

\section{Isoperimetry}

Here we prove an isoperimetric inequality for functions satisfying a certain unimodality criterion. We further show that $(1/(n-1))$-harmonic-concave functions satisfy this unimodality criterion and hence have good isoperimetry.

We begin with a simple lemma that will be used in the proof of the isoperimetric inequality.
\begin{lemma} \label{lemma:interval}
Let $p: [0, 1] \rightarrow \reals_{+}$ be a unimodal function, and let $0 \leq \alpha < \beta \leq 1$. Then
\[
\int_{\alpha}^{\beta} p(t) dt \geq \abs{\alpha - \beta} \min\left\{\int_{0}^{\alpha} p(t) dt, \int_{\beta}^{1} p(t) dt\right\}.
\]
\end{lemma}
\begin{proof} [Proof of Lemma \ref{lemma:interval}(Isoperimetry for 1-dimensional unimodal functions)].
Suppose the maximum of $p$ occurs at $t = t_{\max}$. If $t_{\max} \leq \alpha$ then
\begin{align*}
\int_{\alpha}^{\beta} p(t) dt & \geq p(\beta) \cdot \abs{\alpha - \beta} \geq \abs{\alpha - \beta} \cdot p(\beta) \cdot \abs{\beta - 1} \geq \abs{\alpha - \beta} \cdot \int_{\beta}^{1} p(t) dt
\end{align*}
Otherwise, if $t_{\max} > \alpha$ then
\begin{align*}
\int_{\alpha}^{\beta} p(t) dt & \geq p(\alpha) \cdot \abs{\alpha - \beta} \geq \abs{\alpha - \beta} \cdot p(\alpha) \cdot \abs{\alpha} \geq \abs{\alpha - \beta} \cdot \int_{0}^{\alpha} p(t) dt
\end{align*}
\end{proof}

Now we are ready to prove an isoperimetric inequality for functions satisfying a certain unimodality criterion.
\begin{theorem} \label{thm:iso-unimodal}
Let $f: \reals^{n} \rightarrow \reals_{+}$ be a function whose support has diameter $D$, and $f$ satisfies the following unimodality criterion: For any affine line $L \subseteq \reals^{n}$ and any linear function $l: K \cap L \rightarrow \reals_{+}$, $h(x) = f(x) l(x)^{n-1}$ is unimodal. Let $\reals^{n} = S_{1} \cup S_{2} \cup S_{3}$ be a partition of $\reals^{n}$ into three non-empty subsets. Then
\[
\pi_{f}(S_{3}) \geq \frac{d(S_{1}, S_{2})}{D} \min\left\{\pi_{f}(S_{1}), \pi_{f}(S_{2})\right\}.
\]
\end{theorem}
\begin{proof}
Suppose not. Define $g: \reals^{n} \rightarrow \reals$ and $h: \reals^{n} \rightarrow \reals$ as follows.
\begin{align*}
g(x) = \begin{cases} \dfrac{d(S_{1}, S_{2})}{D} f(x) & \text{if $x \in S_{1}$} \\ 0 & \text{if $x \in S_{2}$} \\ - f(x) & \text{if $x \in S_{3}$} \end{cases}
\quad \text{and} \quad
h(x) = \begin{cases} 0 & \text{if $x \in S_{1}$} \\ \dfrac{d(S_{1}, S_{2})}{D} f(x)
 & \text{if $x \in S_{2}$} \\ - f(x) & \text{if $x \in S_{3}$}.\end{cases}
\end{align*}
Thus
\[
\int_{\reals^{n}} g(x) dx > 0 \quad \text{and} \quad \int_{\reals^{n}} h(x) dx > 0,
\]
Lemma \ref{lemma:local} implies that there exist two points $a, b \in \reals^{n}$ and a linear function $l: [0,1] \rightarrow \reals_{+}$ such that
\begin{align}
\int_{0}^{1} g((1-t)a + tb) l(t)^{n-1} dt > 0 \quad \text{and} \quad \int_{0}^{1} h((1-t)a + tb) l(t)^{n-1} dt > 0. \label{eq:gh}
\end{align}
Moreover, w.l.o.g. we can assume that the points $a$ and $b$ are within the support of $f$, and hence $\norm{a-b} \leq D$. We may also assume that $a\in S_1$ and $b\in S_2$. Consider a partition of the interval $[0,1] = Z_{1} \cup Z_{2} \cup Z_{3}$, where
\[
Z_{i} = \left\{z \in [0,1] \suchthat (1-z)a + zb \in S_{i}\right\}.
\]
For $z_{1} \in Z_{1}$ and $z_{2} \in Z_{2}$, we have
\[
d(S_{1}, S_{2}) \leq d\left(\left(1-z_{1}\right)a + z_{1}b, \left(1-z_{2}\right)a + z_{2}b\right) \leq \abs{z_{1} - z_{2}} \cdot \norm{a - b} \leq \abs{z_{1} - z_{2}} D,
\]
and therefore $d(S_{1}, S_{2}) \leq d(Z_{1}, Z_{2}) D$. Now we can rewrite Equation \eqref{eq:gh} as
\begin{align*}
\int_{Z_{3}} f((1-t)a + tb) l(t)^{n-1} dt & < \frac{d(S_{1}, S_{2})}{D} \int_{Z_{1}} f((1-t)a + tb) l(t)^{n-1} dt \\
& \leq d(Z_{1}, Z_{2}) \int_{Z_{1}} f((1-t)a + tb) l(t)^{n-1} dt \\
\text{and similarly} \\
\int_{Z_{3}} f((1-t)a + tb) l(t)^{n-1} dt & \leq d(Z_{1}, Z_{2}) \int_{Z_{2}} f((1-t)a + tb) l(t)^{n-1} dt
\end{align*}
Define $p: [0,1] \rightarrow \reals_{+}$ as $p(t) = f((1-t)a + tb) l(t)^{n-1}$. From the unimodality assumption in our theorem, we know that $p$ is unimodal. Rewriting the above equations, we have
\begin{align}
\int_{Z_{3}} p(t) dt < d(Z_{1}, Z_{2}) \int_{Z_{1}} p(t) dt \quad \text{and} \quad \int_{Z_{3}} p(t) dt < d(Z_{1}, Z_{2}) \int_{Z_{2}} p(t) dt. \label{eq:interval}
\end{align}
Now suppose $Z_{3}$ is a union of disjoint intervals, i.e., $Z_{3} = \bigcup_{i} (\alpha_{i}, \beta_{i})$, $0\leq \alpha_{1} <\beta _{1} < \alpha _{2} < \beta _{2} <\cdot \cdot \cdot \leq 1$. By Lemma \ref{lemma:interval} we have
\[
\int_{\alpha_{i}}^{\beta_{i}} p(t) dt \geq \abs{\alpha_{i} - \beta_{i}} \cdot \min\left\{\int_{0}^{\alpha_{i}} p(t) dt, \int_{\beta_{i}}^{1} p(t) dt\right\}.\
\]
Therefore, adding these up we get
\begin{align*}
\int_{Z_{3}} p(t) dt & = \sum_{i} \int_{\alpha_{i}}^{\beta_{i}} p(t) dt \\
& \geq \abs{\alpha_{i} - \beta_{i}} \cdot \sum_{i} \min\left\{\int_{0}^{\alpha_{i}} p(t) dt, \int_{\beta_{i}}^{1} p(t) dt\right\} \\
& \geq d(Z_{1}, Z_{2}) \cdot \min\left\{\int_{Z_{1}} p(t) dt, \int_{Z_{2}} p(t) dt\right\}.
\end{align*}
since there must be some $i$ such that $Z_{1}$ and $Z_{2}$ are separated by the interval $(\alpha_{i}, \beta_{i})$. But then we get a contradiction to Equation \eqref{eq:interval}. This completes the proof of Theorem \ref{thm:iso-unimodal}.\\
\end{proof}

\subsection{Isoperimetry of (1/(n-1))-Harmonic-concave functions}
We show that $(1/(n-1))$-harmonic-concave functions satisfy the unimodality criterion used in the proof of Theorem \ref{thm:iso-unimodal}. Therefore, as a corollary, we get an isoperimetric inequality for $(1/(n-1))$-harmonic-concave functions, which is a subclass of harmonic-concave functions.

\begin{prop} \label{prop:harmodal}
Let $f: \reals^{n} \rightarrow \reals_{+}$ be a smooth $(1/(n-1))$-harmonic-concave function and $l: [0, 1] \rightarrow \reals_{+}$ be a linear function. Now let $a, b \in \reals^{n}$ and define $h: [0, 1] \rightarrow \reals_{+}$ as $h(t) = f((1-t)a + tb) l(t)^{n-1}$. Then $h$ is a unimodal function.
\end{prop}
\begin{proof} [Proof of Proposition \ref{prop:harmodal} (Unimodality of $1/(n-1)$-harmonic-concave functions)]
Define $g(x) = \left(1/f(x)\right)^{1/(n-1)}$. Since $f$ is $(1/(n-1))$-harmonic-concave, $g$ is convex and we can rewrite $h$ as
\[
h(t) = f((1-t)a + tb) l(t)^{n-1} = \left(\frac{l(t)}{g((1-t)a + tb)}\right)^{n-1} = \left(\frac{l(t)}{q(t)}\right)^{n-1},
\]
where $q(t) = g((1-t)a + tb)$ which is also convex. Also w.l.o.g. we can do a linear transformation that transforms $l$ into the identity function $l(t) = t$ without affecting the convexity of $q$. Thus, it suffices to show that if $h(t) = \left(t/q(t)\right)^{1/(n-1)}$, where $q$ is a convex function, then $h$ is unimodal. Indeed,
\begin{align*}
\frac{d}{dt} h(t)^{n-1} & = \frac{q(t) - t \dfrac{d}{dt}q\left(t\right)}{q(t)^{2}} \\
\frac{d^{2}}{dt^{2}} h(t)^{n-1} & = \frac{q(t)\left(2t \left(\dfrac{d}{dt}q\left(t\right)\right)^{2} - 2 q(t) \dfrac{d}{dt}q\left(t\right) - t q(t) \dfrac{d^{2}}{dt^{2}}q\left(t\right)\right)}{q(t)^{4}}.
\end{align*}
If there exists a local optimum for $h^{n-1}$ at $t = t_{0}$ then
\[
\frac{d}{dt} h(t_{0})^{n-1} = 0 \Rightarrow q(t_{0}) = t_{0} \dfrac{d}{dt}q\left(t_{0}\right) \Rightarrow \frac{d^{2}}{dt^{2}} h(t_{0})^{n-1} = \frac{- t_{0} q(t_{0})^{2}\dfrac{d^{2}}{dt^{2}}q\left(t_{0}\right)}{q(t_{0})^{4}} \leq 0,
\]
because $t_{0} \in (0, 1)$, and $\dfrac{d^{2}}{dt^{2}} q\left(t\right)\geq 0$ as $q$ is convex. This implies that every local optimum is a local maximum and hence there are no convex regions in $h(t)^{n-1}$. Hence, there are no convex regions in $h(t)$ which implies that $h(t)$ is unimodal.\\
\end{proof}

We get Theorem \ref{thm:isoperimetric} as a corollary of Theorem \ref{thm:iso-unimodal} and Proposition \ref{prop:harmodal}.
\subsection{Lower bound for isoperimetry}
In this section, we show that $(1/(n-1))$-harmonic-concave functions are the limit of isoperimetry by showing a $(1/(n-1-\eps))$-harmonic concave function with poor isoperimetry for $0<\eps\leq 1$.

\begin{proof} [Proof of Theorem \ref{thm:limitIsoperimetry}]
The proof is based on the following construction.
Consider $K \subseteq \reals^{n}$ defined as follows.
\[
K = \left\{x \suchthat 0 \leq x_{1} < \frac{1}{1+\delta}~ \text{and}~ x_{2}^{2} + x_{3}^{2} + \dotsc + x_{n}^{2} \leq (1-x_{1})^{2}\right\},
\]
where $\delta > 0$. $K$ is a parallel section of a cone symmetric around the $X_{1}$-axis and is therefore convex. Now we define a function $f: \reals^{n} \rightarrow \reals_{+}$ whose support is $K$.
\[
f(x) = \begin{cases} \dfrac{C}{\left(1 - (1+\delta)x_{1}\right)^{n-1-\eps}} & \quad \text{if $x \in K$,} \\ 0 & \quad \text{if $x \notin K$}, \end{cases}
\]
where $C$ is the appropriate constant so as to make $\pi_{f}(K) = 1$. By definition, $f$ is a $1/(n-1-\eps)$-harmonic concave function.

Define a partition $\reals^{n} = S \cup T$ as $S = \{x \in K \suchthat 0 \leq x_{1} \leq t\}$ and $T = \reals^{n} \setminus S$. The theorem holds for a suitable choice of $t$.
[Proof of Theorem \ref{thm:limitIsoperimetry} (Limit of isoperimetry)]

We want to show that
\[
\frac{\pi_{f} (\partial S)}{\min\{\pi_{f}(S), \pi_{f}(T)\}} = O(c^{-n}),
\]
for some constant $c > 1$, which means that $f$ does not satisfy the isoperimetric inequality. (By abuse of notation, we use $\pi_{f}(\partial S)$ for the area measure defined by $f$ on the boundary of $S$.) In order to achieve this, it seems better for us to choose a value of $t$ that minimizes $\pi_{f}(\partial S)$.
\[
\pi_{f}(\partial S) = V_{n-1} C (1-t)^{n-1} \left(1-(1+\delta)t\right)^{1+\eps-n},
\]
where $V_{n}$ denotes the volume of a unit ball in $\reals^{n}$. Simple calculus shows that $\pi_{f}(\partial S)$ is convex as a function of $t$ over $[0, 1/(1+\delta)]$ and attains the minimum when
\[
t = t_{\min} = \frac{1}{1+\delta} - \frac{n-1-\eps}{\eps} \cdot \frac{\delta}{1+\delta}.
\]
Moreover, $\pi_{f}(\partial S)$ is decreasing for $0 \leq t \leq t_{\min}$ and increasing for $t_{\min} \leq t \leq 1/(1+\delta)$. Thus, using $t = t_{\min}$ to define the partition $\reals^{n} = S \cup T$, we get
\begin{align*}
\pi_{f}(\partial S) & = V_{n-1} C (1-t_{\min})^{n-1} \left(1-(1+\delta)t_{\min}\right)^{1+\eps-n} \\
& = V_{n-1} C \left(\frac{\delta}{1+\delta}\right)^{n-1} \left(\frac{n-1-\eps}{\eps}\right)^{n-1} \left(\frac{\delta (n-1-\eps)}{\eps}\right)^{1+\eps-n}.
\end{align*}
and
\begin{align*}
\pi_{f}(S) & = \int_{S} f(x) dx \\
& = V_{n-1} C \int_{0}^{t_{\min}} (1-x_{1})^{n-1} \left(1-(1+\delta)x_{1}\right)^{1+\eps-n} dx_{1} \\
& \geq V_{n-1} C \int_{0}^{t'} (1-x_{1})^{n-1} \left(1-(1+\delta)x_{1}\right)^{1+\eps-n} dx_{1} \\
& \qquad \qquad \text{where $0 < t' = \frac{1}{2(1+\delta)} < t_{\min}$} \\
& \geq V_{n-1} C (1-t')^{n-1} \left(1-(1+\delta)t'\right)^{1+\eps-n} t' \\
& = V_{n-1} C \left(\frac{1}{2} + \frac{\delta}{2(1+\delta)}\right)^{n-1} \left(\frac{1}{2}\right)^{1+\eps-n} \frac{1}{2(1+\delta)} \\
& \geq V_{n-1} C \frac{1}{2^{\eps}} \cdot \frac{1}{2(1+\delta)}.
\end{align*}
and
\begin{align*}
\pi_{f}(T) & = \int_{T} f(x) dx \\
& = V_{n-1} C \int_{t_{\min}}^{1/(1+\delta)} (1-x_{1})^{n-1} \left(1-(1+\delta)x_{1}\right)^{1+\eps-n} dx_{1} \\
& \geq V_{n-1} C \int_{t''}^{1/(1+\delta)} (1-x_{1})^{n-1} \left(1-(1+\delta)x_{1}\right)^{1+\eps-n} dx_{1} \\
& \quad \quad \text{where $t_{\min} < t'' = \frac{1}{1+\delta} - \frac{1}{n^{2}} \cdot \frac{n-1-\eps}{\eps} \cdot \frac{\delta}{1+\delta} < \frac{1}{1+\delta}$} \\
& \geq V_{n-1} C (1-t'')^{n-1} \left(1-(1+\delta)t''\right)^{1+\eps-n} \left(\frac{1}{1+\delta} - t''\right) \\
& = V_{n-1} C \left(\frac{\delta}{1+\delta}\right)^{n-1} \left(1+\frac{n-1-\eps}{n^{2}\eps}\right)^{n-1} \left(\frac{\delta(n-1-\eps)}{n^{2}\eps}\right)^{1+\eps-n} \frac{\delta(n-1-\eps)}{n^{2}\eps(1+\delta)} \\
& \geq V_{n-1} C \left(\frac{\delta}{1+\delta}\right)^{n-1} \left(\frac{(n-1)(\eps(n+1)+1)}{n^2\eps} \right)^{n-1} \left(\frac{\delta(n-1-\eps)}{n^{2}\eps}\right)^{1+\eps-n} \frac{\delta(n-1-\eps)}{n^{2}\eps(1+\delta)}
\end{align*}
Therefore,
\begin{align*}
\frac{\pi_{f}(\partial S)}{\pi_{f}(S)} & \leq \left(\frac{\delta}{1+\delta}\right)^{n-1} \left(\frac{n-1-\eps}{\eps}\right)^{n-1} \left(\frac{\delta (n-1-\eps)}{\eps}\right)^{1+\eps-n} 2^{1+\eps} (1+\delta) \\
& \leq \delta^{\eps} \left(\frac{1}{1+\delta}\right)^{n-1} (n-1-\eps)^{\eps} \left(\frac{1}{\eps}\right)^{\eps} 2^{1+\eps} (1+\delta) \\
& \leq Cn\left((1+\eps)^{-\eps n}\right) &\quad \text{(Using $\delta = 1/({1+\eps})^{n}$)}
\end{align*}
for some constant $C>0$ and
\begin{align*}
\frac{\pi_{f}(\partial S)}{\pi_{f}(T)} & \leq \left(\frac{\eps n^{2(1+\eps)} (1+\delta)}{n-1-\delta}\right) \left(\frac{1}{\delta(\eps(n+1)+1)^{n-1}}\right)\\
&\leq Cn^6\left(\frac{1}{(1+\frac{\eps}{1+\eps}n)^n}\right)&\quad \text{(Using $\delta = 1/({1+\eps})^{n}$)}\\
&\leq C2^{-n}
\end{align*}
for some constant $C>0$.
Putting these together we have
\[
\frac{\pi_{f}(\partial S)}{\min\left\{\pi_{f}(S), \pi_{f}(T)\right\}} \leq Cn((1+\eps)^{-\eps n}),
\]
for some constant $C>0$.
\end{proof}

\section{Sampling s-Harmonic-concave functions}
Throughout this section, let $f:\reals^n\rightarrow \reals_+$ be an $s$-harmonic-concave function given by an oracle with parameters $(\alpha,\delta)$ such that $s\leq1/(n-1)$. Let $K$ be the convex set over which we want to sample points according to $f$. We also assume that $K$ contains a ball of radius $\delta$.
We state a technical lemma related to the parameters and the harmonic-concavity of the function.
\begin{lemma}\label{lem:parameters_hc}
Suppose $f:\reals^n\rightarrow \reals$ is a $s$-harmonic-concave function with parameters $(\alpha,\delta)$ as defined earlier. For any constant $c$ such that $1<c<\alpha$, if $\norm{x-z}\leq \frac{cs\delta}{\alpha^s-1}$, then $\frac{f(x)}{f(z)}\leq c$.
\end{lemma}
\begin{proof} [Proof of \ref{lem:parameters_hc} (Modifying parameters)]
Let $x,y\in \reals^n$ such that $\norm{x-y}\leq \delta$ and let $z=(1-t)x+ty$ for some $t\in (0,1)$. Then, by the $s$-harmonic-concavity, we have that
\begin{align*}
f(z)&\geq \left(\frac{1-t}{f(x)^s}+\frac{t}{f(y)^s}\right)^{-\frac{1}{s}}\\
&\geq \left(\frac{1-t}{f(x)^s}+\frac{t\alpha^s}{f(x)^s}\right)^{-\frac{1}{s}} \quad \quad \quad \quad \text{(Since $\norm{x-y}\leq \delta$ implies $f(x)/f(y)<\alpha$)}\\
&= \frac{f(x)}{(1+t(\alpha^s-1))^{\frac{1}{s}}}
\end{align*}
Since, $\norm{x-z}\leq \frac{cs\delta}{\alpha^s-1}$, we get the desired conclusion.
\end{proof}
Hence, for the $s$-harmonic-concave functions with parameters $(\alpha,\delta)$, the above lemma states that they also have parameters $(c,\frac{cs\delta}{(\alpha^{s}-1)})$ for any constant $c$ such that $1<c<\alpha$. In particular, if $\alpha>2$, this suggests that we may use $(2,\frac{2s\delta}{(\alpha^{s}-1)})$ as the parameters and if $\alpha\leq 2$, then we may use $(2,\delta)$ as the parameters. Thus, the parameters are $(2,\min\{\delta,\frac{2s\delta}{(\alpha^{s}-1)}\})$.

In order to sample, we need to show that $K$ contains points of good local conductance. For this, define
\[
K_{r} = \left\{x \in K \suchthat \frac{\vol{\ball{x}{r} \cap K}}{\vol{\ball{x}{r}}} \geq \frac{3}{4}\right\}.
\]
The idea is that, for appropriately chosen $r$, the log-lipschitz-like constraint will enforce that the points in $K_r$ have good local conductance. Further, we have that the measure in $K_r$ is close to the measure of $f$ in $K$ based on the radius $r$.
\begin{lemma}\label{lem:stepSize_hc}
For any $r > 0$, the set $K_r$ is convex and
\[
\pi_{f}(K_{r}) \geq 1 - \frac{4r \sqrt{n}}{\delta}.
\]
\end{lemma}
\begin{proof} [Proof of \ref{lem:stepSize_hc} (Points of good local conductance)]
Now, the convexity of $K_r$ is a direct consequence of the Brunn-Minkowski inequality: for compact sets $A, B$ and their Minkowski sum $A+B$, $\vol{A+B}^{1/n} \ge \vol{A}^{1/n}+ \vol{B}^{1/n}$. To prove the second part, we need the following lemma  paraphrased from \cite{KLS97}.
\begin{lemma}\label{lem:blowup}
Let $K$ be a convex set containing a ball of radius $t$.
Then
$\int_{x \in K} \int_{y \in (x + \ball{}{r}) \setminus K} dydx \leq \frac{r\sqrt{n}}{2t}vol(K)vol(\ball{}{r}).$
\end{lemma}

The target density $f$ (with support $K$) can be viewed as follows: first we pick a level set $L(t)$ of $f$, with the appropriate marginal distribution on $t$, then we pick x uniformly from $L(t)$.

Now consider any level set $L(t)$ where $t$ is picked according to the appropriate marginal distribution. If $L(t)$ does not contain a ball of radius $\delta/2$, then the probability of stepping out of $K$ from any point in $L(t)$ is 0. If $L(t)$ contains such a ball, the probability of stepping out is bounded above using Lemma \ref{lem:blowup} by $r\sqrt{n}/\delta$. That is,
\[
\frac{\vol{\ball{x}{r}\setminus K}}{\vol{\ball{}{r}}}\leq \frac{r\sqrt{n}}{\delta}
\]
where $x$ is a random point in $L(t)$. Consider a random variable $g(x)=\frac{\vol{\ball{x}{r}\setminus K}}{\vol{\ball{}{r}}}$ when $x$ is a random point in $L(t)$. Since by the above inequality $E(g(x))\leq r\sqrt{n}/\delta$, we have that
\[
\prob{g(x)>\frac{4r\sqrt{n}}{\delta}}\leq \frac{1}{4}
\]
This implies that at most $1/4$-th of the fraction of points in $L(t)$ step out of $K$ with probability greater than $4r\sqrt{n}/\delta$. Hence, at least $3/4$-th of the fraction of points in $L(t)$ step out of $K$ with probability at most $4r\sqrt{n}/\delta$.

That is, $3/4$-th of the fraction of points in $L(t)$ remain within $K$ with probability at least $1-(4r\sqrt{n}/\delta)$.

Hence $\pi_f(K_r)\geq 1-(4r\sqrt{n}/\delta)$.
\end{proof}

\subsection{Coupling}
In order to prove conductance, we need to prove that when two points are geometrically close, then their one-step distributions overlap. We will need the following technical lemma about spherical caps to prove this.
\begin{lemma}\label{lem:sphereHalfSpace}
Let $H$ be a halfspace in $\reals ^n$ and $B_x$ be a ball whose center is at a distance at most $tr/\sqrt{n}$ from $H$. Then
\[
e^{-\frac{t^2}{4}}>\frac{2\vol{H\cap \ball{}{r}}}{\vol{\ball{}{r}}}>1-t
\]
\end{lemma}

\begin{lemma}\label{lem:coupling_hc}
For $r\leq \min\{\delta,\frac{2s\delta}{(\alpha^{s}-1)}\}$, if $u, v \in K_r$, $\norm{u-v} < r/16\sqrt{n}$, then
\[
d_{tv} (P_{u}, P_{v})\leq 1-\frac{7}{16}
\]
\end{lemma}
\begin{proof} 
We may assume that $f(v)\geq f(u)$. Then,
\[
d_{tv} (P_{u}, P_{v})  \leq 1 - \frac{1}{\vol{\ball{}{r}}} \int_{\ball{v}{r} \cap \ball{u}{r}\cap K} \min\left\{1, \frac{f(y)}{f(v)}\right\} dy
\]
Let us lower bound the second term in the right hand side.
\begin{align*}
\int_{\ball{v}{r} \cap \ball{u}{r}\cap K} \min\left\{1, \frac{f(y)}{f(v)}\right\} dy
& \geq \int_{\ball{v}{r} \cap \ball{u}{r}\cap K} \min\left\{1, \frac{f(y)}{f(v)}\right\} dy \\
& \geq \left(\frac{1}{2}\right) \vol{\ball{v}{r} \cap \ball{u}{r} \cap  K} \quad \quad \text{(Consequence of Lemma \ref{lem:parameters_hc})}\\
& \geq \left(\frac{1}{2}\right) \left(\vol{\ball{v}{r}} - \vol{\ball{v}{r}\setminus\ball{u}{r}}-\vol{\ball{v}{r}\setminus K}\right)\\
& \geq \left(\frac{1}{2}\right) \left(\vol{\ball{v}{r}} - \frac{1}{16}\vol{\ball{}{r}}- \frac{1}{16}\vol{\ball{}{r}}\right)\\
& \geq \left(\frac{7}{16}\right)\vol{\ball{}{r}}
\end{align*}
where the bound on $\vol{\ball{v}{r}\setminus\ball{u}{r}}$ is derived from Lemma \ref{lem:sphereHalfSpace} and $\vol{\ball{v}{r}\setminus K}$ is bounded using the fact that $v\in K_r$. Hence,
\[
d_{tv} (P_{u}, P_{v})\leq 1-\frac{7}{16}
\]
\end{proof}

\subsection{Conductance and mixing time}
Consider the ball walk with metropolis filter using the $s$-harmonic-concave distribution function oracle (whose parameters are $(\alpha,\delta)$) with ball steps of radius $r$.
\begin{lemma} \label{lemma:ergo_hc}
For any $\eps_1>0$, let $D$ be the diameter of the ball such that $\pi_f(\ball{0}{D})\geq 1-\frac{\eps_1}{2}$. Let $S \subseteq \reals^{n}$ be such that $\pi_f(S)\geq \eps_1$ and $\pi_f(\reals^n\backslash S)\geq \eps_1$. Then, for ball walk radius $r\leq \min\{\frac{\eps_1\delta}{8\sqrt{n}},\frac{2s\delta}{(\alpha^{s}-1)}\}$, we have that
\[
\Phi(S) \geq \frac{r}{2^{9}\sqrt{n}D}\min\{\pi_f(S)-\eps_1,\pi_f(\reals^n\backslash S)-\eps_1\}
\]
\end{lemma}
\begin{proof} [Proof of Lemma \ref{lemma:ergo_hc}]
Given $K$ as the convex set over which we want to sample points according to the $(1/(n-1))$-harmonic-concave distribution, define $K_r$ as above. Further, for all subsets $A\subseteq K$, we have that
\begin{align}\label{eqn:pifBound_hc}
\pi_f(A\cap {K'}) & = \pi_f(A\cap \cap K_r) \nonumber\\
& \geq \pi_f(A) - \eps_1 & \text{(by Lemma \ref{lem:stepSize_hc} using $r\leq \frac{\eps_1\delta}{8\sqrt{n}}$)}
\end{align}
Using $S$, define $S_{1}, S_{2}, S_{3}$ as follows.
\begin{align*}
S_{1} & = \left\{x \in S \suchthat P_{x}(\reals^{n} \setminus S) \leq \frac{7}{32} \right\} \\
S_{2} & = \left\{x \in \reals^{n} \setminus S \suchthat P_{x}(S) \leq \frac{7}{32} \right\} \\
S_{3} & = \reals^{n} \setminus (S_{1} \cup S_{2}).
\end{align*}
Also define $S_{i}'$ as
\[
S_{i}' = S_{i} \cap K_r, \qquad \text{for $i=1, 2, 3$.}
\]
The ergodic flow of $S$ can be written as
\begin{align*}
\Phi(S) & = \frac{1}{2} \left( \int_{S} P_{x}(\reals^{n} \setminus S) f(x)~ dx + \int_{\reals_{n} \setminus S} P_{x}(S) f(x)~ dx \right) \\
& \geq \frac{7}{64} \pi_{f}(S_{3}) \geq \frac{1}{16} \pi_f(S_3')
\end{align*}

Suppose $\pi_{f}(S_{1}') \leq \pi_{f}(S \cap {K_r})/2$, then $\pi_{f}(S_3') \geq \pi_{f}(S \cap {K_r})/2$ because $S \cap {K_r} \subseteq S_{1}' \cup S_3'$. Thus
\begin{align*}
\Phi(S) & \geq \frac{1}{32} \pi_{f}(S \cap {K_r}) \geq \frac{1}{32} \left(\pi_{f}(S) - \eps_1\right) &\text{(By (\ref{eqn:pifBound_hc}))}
\end{align*}
which implies the lemma.
So, we may assume that $\pi_{f}(S_{1}') \geq \pi_{f}(S \cap {K_r})/2$, and similarly $\pi_{f}(S_{2}') \geq \pi_{f}\left((\reals^{n} \setminus S) \cap {K_r}\right)/2$. Then, using Theorem \ref{thm:isoperimetric},
\begin{align*}
\Phi(S) & \geq \frac{1}{16} \pi_{f}(S_3') \\
& \geq \frac{1}{16} \frac{d(S_{1}',S_{2}')}{\diam} \min\left\{\pi_{f}(S_{1}'), \pi_{f}(S_{2}')\right\} \\
& \geq \frac{1}{32} \frac{d(S_{1}, S_{2})}{\diam} \cdot \min\left\{\pi_{f}(S \cap {K_r}), \pi_{f}\left((\reals^{n} \setminus S) \cap {K_r}\right)\right\} \\
\end{align*}
Now, for any $u\in S_1$, $v\in S_2$, $d_{tv}(P_u,P_v)\geq 1- 2\cdot \max\{P_u(\reals ^n\setminus S),P_v(S)\}\geq 1-\frac{7}{16}$. Also, $r \leq \min\{\frac{\eps_1\delta}{8\sqrt{n}},\frac{2s\delta}{(\alpha^{s}-1)}\}< \min\{\delta,\frac{2s\delta}{(\alpha^{s}-1)}\}$. Hence, by Lemma \ref{lem:coupling_hc}, $d(S_1,S_2)\geq \frac{r}{16\sqrt{n}}$. Therefore,
\begin{align*}
\Phi(S) & \geq \frac{1}{32} \cdot \frac{r}{16 \sqrt{n}} \cdot \frac{1}{\diam} \cdot \min\left\{\pi_{f}(S) - \eps_1, \pi_{f}(\reals^{n} \setminus S) - \eps_1\right\}  &\text{By (\ref{eqn:pifBound_hc})}\\
& \geq \frac{r}{2^{9} \sqrt{n} {\diam} }\min\left\{\pi_{f}(S) - \eps_1, \pi_{f}(\reals^{n} \setminus S) - \eps_1\right\}
\end{align*}
\end{proof}

Using the above lemma, we prove Theorem \ref{thm:sampling_hc}.
\begin{proof} [Proof of Theorem \ref{thm:sampling_hc}]
On setting $\eps_1=\eps/2H$ in Lemma \ref{lemma:ergo_hc}, we have that for ball-walk radius $r= \min\{\frac{\eps\delta}{16H\sqrt{n}},\frac{2s\delta}{(\alpha^{s}-1)}\}$,
\[
\phi_{\eps_1}\geq\frac{r}{2^{9}\sqrt{n}D}.
\]
By definition $H_s\leq H\cdot s$ and hence by Lemma \ref{lemma:mix},
\[
\abs{\sigma_{m}(S) - \pi_{f}(S)} \leq H\cdot {s} + {H}\cdot exp\left\{-\frac{mr^2}{2^{19}nD^2}\right\}
\]
which gives us that beyond
\[
m\geq \frac{2^{19}nD^2}{r^2}\log{\frac{2H}{\eps}}
\]
steps, $\abs{\sigma_{m}(S) - \pi_{f}(S)} \leq\eps$. Substituting for $r$, we get the theorem.
\end{proof}

\subsection{Sampling the Cauchy density}
In this section, we prove certain properties of the Cauchy density along with the crucial coupling lemma leading to Theorem \ref{thm:cauchy_sampling}. Without loss of generality, we may assume that the distribution given by the oracle is,
\begin{equation}
f(x)\propto \left\{
\begin{array}{ll}
1/(1+||x||^2)^{\frac{n+1}{2}} & \text{if } x\in K \text{,}\\
0 & \text{otherwise.}
\end{array}\right.
\end{equation}
This is because, either we are explicitly given the matrix $A$ of a general Cauchy density, or we can compute it using the function $f$ at a small number of points and apply a linear transformation. Further, note that by the hypothesis of Theorem \ref{thm:cauchy_sampling}, we may assume that $K$ contains a unit ball.
\begin{prop}\label{prop:cauchy_hc}
The Cauchy density function is $(1/(n-1))$-harmonic-concave.
\end{prop}
\begin{proof} [Proof of Proposition \ref{prop:cauchy_hc} (Cauchy is 1/(n-1)-harmonic-concave)]
To check for $(1/(n-1))$-harmonic-concavity, we need to check that
\begin{displaymath}
g(x)=(1+\sum_i x_i ^2)^{\frac{n+1}{2(n-1)}}
\end{displaymath}
is convex. This follows from the following two observations:
\begin{enumerate}
\item $\frac{n+1}{2(n-1)}>\frac{1}{2}$
\item $(1+\sum_i x_i ^2)^{1/2}$ is convex.
\end{enumerate}
\end{proof}

Proposition \ref{prop:large-ball} says that we can find a ball of radius $O(\sqrt{n}/\eps_1)$ outside which the Cauchy density has at most $\eps_1$ probability mass.
\begin{prop}\label{prop:large-ball}
\[
\prob{\norm{x} \geq \frac{2\sqrt{2n}}{ \eps_1}} \leq \eps_1.
\]
\end{prop}
\begin{proof} [Proof of Proposition \ref{prop:large-ball} (Ball of large mass)]
\begin{align*}
\prob{\norm{x} \geq t} & = \dfrac{\int_{\norm{x} \geq t} \dfrac{1}{(1 + \sqnorm{x})^{(n+1)/2}}~ dx}{\int_{\reals^{n}} \dfrac{1}{(1 + \sqnorm{x})^{(n+1)/2}}~ dx} \\
& = \dfrac{\int_{\sph^{n-1}} d\omega \cdot \int_{t}^{\infty} \dfrac{r^{n-1}}{(1 + r^{2})^{(n+1)/2}}~ dr}{\int_{\sph^{n-1}} d\omega \cdot \int_{0}^{\infty} \dfrac{r^{n-1}}{(1 + r^{2})^{(n+1)/2}}~ dr} & \text{using polar coordinates} \\
& = \dfrac{\int_{t}^{\infty} \dfrac{r^{n-1}}{(1 + r^{2})^{(n+1)/2}}~ dr}{\int_{0}^{\infty} \dfrac{r^{n-1}}{(1 + r^{2})^{(n+1)/2}}~ dr}.
\end{align*}
Now we will analyze the numerator, call it $N$, and the denominator, call it $D$, separately.
\begin{align} \label{eq:num}
N = \int_{t}^{\infty} \frac{r^{n-1}}{(1 + r^{2})^{(n+1)/2}}~ dr \leq \int_{t}^{\infty} \frac{1}{r^{2}}~ dr = \frac{1}{t},
\end{align}
and
\begin{align} \label{eq:denom}
D & = \int_{0}^{\infty} \frac{r^{n-1}}{(1 + r^{2})^{(n+1)/2}}~ dr \nonumber \\
& = \int_{0}^{\pi/2} \frac{\sin^{n-1} \theta}{\cos^{n-1} \theta}~ \frac{1}{\sec^{n+1} \theta}~ \sec^{2} \theta~ d\theta & \text{using $r = \tan \theta$} \nonumber \\
& = \int_{0}^{\pi/2} \sin^{n-1} \theta~ d\theta \nonumber \\
& = \left[ - \frac{\sin^{n-2} \theta \cos \theta}{n-1} \right]_{0}^{\pi/2} + \frac{n-2}{n-1} \int_{0}^{\pi/2} \sin^{n-3} \theta~ d\theta \nonumber \\
& = \frac{n-2}{n-1} \int_{0}^{\pi/2} \sin^{n-3} \theta~ d\theta \nonumber \\
& = \begin{cases} \dfrac{(n-2)!!}{(n-1)!!}, \quad \text{if $n$ is even}. \nonumber \\ \dfrac{\pi}{4} \cdot \dfrac{(n-2)!!}{(n-1)!!}, \quad \text{if $n$ is odd}. \end{cases} & \text{by continuing the recursion} \nonumber \\
& \geq \frac{\pi}{4} \cdot \frac{(n-2)!}{(n-1)!} \nonumber \\
& \geq \frac{\pi}{4} \cdot \frac{\sqrt{2}}{\pi\sqrt{n}} \nonumber & \text{by Wallis' inequality}\\
& = \frac{1}{2\sqrt{2n}}.
\end{align}
Therefore, from Equations \eqref{eq:num} and \eqref{eq:denom} we get
\[
\prob{\norm{x} \geq t} \leq \frac{N}{D} \leq \frac{2\sqrt{2n}}{t}.
\]
This implies
\[
\prob{\norm{x} \geq \frac{2\sqrt{2n}}{ \eps_1}} \leq \eps_1.
\]
\end{proof}

Proposition \ref{prop:smooth} shows the smoothness property of the Cauchy density. This is the crucial ingredient used in the stronger coupling lemma. Define $K_r$ as before. Then,

\begin{prop} \label{prop:smooth}
For $x \in K_r$, let
\[
C_x = \{y \in \ball{x}{r}: |x\cdot(x-y)|\leq \frac{4r||x||}{\sqrt{n}} \}
\]
and $y\in C_x$. Then,
\[
\frac{f(x)}{f(y)} \geq 1-4r\sqrt{n}
\]
\end{prop}
\begin{proof} [Proof of Proposition \ref{prop:smooth} (Smoothness property)]
We have that,
\begin{align*}
||y||^2 & = ||x||^2 + ||y-x||^2 + 2x\cdot(x-y)\\
&\geq ||x||^2  - \frac{8r||x||}{\sqrt{n}}
\end{align*}
Therefore
\begin{align*}
\frac{f(x)}{f(y)} &= \left(\frac{1+||y||^2}{1+||x||^2}\right)^{\frac{n+1}{2}}\\
&\geq \left(1-\frac{8r||x||}{\sqrt{n}(1+||x||^2)}\right)^{\frac{n+1}{2}} \\
&\geq \left(1-\frac{4r}{\sqrt{n}}\right)^{\frac{n+1}{2}} &\quad \quad \quad \text{$\left(max\ \frac{t}{1+t^2} = \frac{1}{2}\right)$}\\
&\geq 1-4r\sqrt{n}
\end{align*}
\end{proof}

Finally, we have the following coupling lemma.
\begin{lemma}\label{lem:gap}
For $r\leq 1/\sqrt{n}$, if $u, v \in K_{r}$, $\norm{u-v} < r/16\sqrt{n}$, then
\[
d_{tv} (P_{u}, P_{v}) < \frac{1}{2}.
\]
\end{lemma}
\begin{proof} [Proof of Lemma \ref{lem:gap} (Coupling lemma for Cauchy)]
We may assume that $f(v)\geq f(u)$. Then,
\[
d_{tv} (P_{u}, P_{v})  \leq 1 - \frac{1}{\vol{\ball{}{r}}} \int_{\ball{u}{r} \cap \ball{v}{r}\cap K} \min\left\{1, \frac{f(y)}{f(v)}\right\} dy
\]
Let us lower bound lower bound the second term in the right hand side.
\begin{align*}
\int_{\ball{u}{r} \cap \ball{v}{r}\cap K} \min\left\{1, \frac{f(y)}{f(v)}\right\} dy
& \geq \int_{\ball{v}{r} \cap \ball{u}{r}\cap K \cap  C_v} \min\left\{1, \frac{f(y)}{f(v)}\right\} dy \\
& \geq (1-4r\sqrt{n}) \vol{\ball{v}{r} \cap \ball{u}{r} \cap C_v \cap  K} \quad \quad \quad \quad \quad \quad \text{(by Lemma \ref{prop:smooth})}\\
& \geq (1-4r\sqrt{n}) \left(\vol{\ball{v}{r}} - \vol{\ball{v}{r}\setminus\ball{u}{r}}-\vol{\ball{v}{r}\setminus C_v}-\vol{\ball{v}{r}\setminus K}\right)\\
& \geq (1-4r\sqrt{n}) \left(\vol{\ball{v}{r}} - \frac{1}{16}\vol{\ball{}{r}}-\frac{1}{16}\vol{\ball{}{r}}-e^{-4}\vol{\ball{}{r}}\right)\\
& \geq (1-4r\sqrt{n}) \left( \frac{13}{16}\right)\vol{\ball{}{r}}
\end{align*}
where the bounds on $\vol{\ball{v}{r}\setminus\ball{u}{r}}$ and $\vol{\ball{v}{r}\setminus C_v}$ are derived from Lemma \ref{lem:sphereHalfSpace} and $\vol{\ball{v}{r}\setminus K}$ is bounded using the fact that $v\in K_r$. Since $r\leq 1/\sqrt{n}$,
\[
d_{tv} (P_{u}, P_{v})\leq \frac{3+4r\sqrt{n}}{16} \leq \frac{1}{2}
\]
\end{proof}

The proof of conductance and mixing bound follow the proof of mixing bound for $s$-harmonic-concave functions closely. Comparing the above coupling lemma with that of $s$-harmonic-concave functions (Lemma \ref{lem:coupling_hc}), we observe that the improvement is obtained due to the constraint on the radius of the ball walk in the coupling lemma. In the case of Cauchy, a slightly relaxed radius suffices for points close to each other to have a considerable overlap in their one-step distribution.
\subsection{Discussion}
There are two aspects of our algorithm and analysis that merit improvement. The first is the dependence on the diameter, which could perhaps be made logarithmic by applying an appropriate affine transformation as in the case of logconcave densities. The second is eliminating the dependence on the smoothness parameter entirely, by allowing for sharp changes locally and considering a smoother version of the original function. Both these aspects seem to be tied closely to proving a tail bound on a $1$-dimensional marginal of an $s$-harmonic concave function.
\bibliographystyle{amsalpha}
\bibliography{ball-walk-ref}
\section{Appendix: proofs}
\subsection{Sampling the Cauchy density}
The following lemma gives the parameters for the Cauchy distribution function.
\begin{prop} \label{prop:smooth_const}
If $\norm{y-x} \leq 1/n$, then
\[
\frac{f(y)}{f(x)} \geq \frac{1}{e},
\]
\end{prop}
\begin{proof}
Let $\nabla_{u}$ denote the directional derivative of a function along $u$. Then
\begin{align}
\log \frac{f(x)}{f(y)} & = \log f(x) - \log f(y) \nonumber \\
& \leq \sup_{z \in \reals^{n}} \sup_{\norm{u}=1} \norm{\nabla_{u} \log f(z)} \cdot \norm{x-y} \nonumber \\
& = \sup_{z \in \reals^{n}} \sup_{\norm{u}=1} \frac{1}{f(z)} \norm{\nabla_{u} f(z)} \cdot \norm{x-y} \nonumber \\
& = \norm{x-y} \sup_{z \in \reals^{n}} \frac{1}{f(z)} \sup_{\norm{u}=1} \norm{\nabla_{u} f(z)} \label{eq:log-ratio}.
\end{align}
We know that by definition
\[
f(z) = \frac{c}{\left(1+\sqnorm{z}\right)^{(n+1)/2}},
\]
for some constant $c$ that includes $\det(A)^{-1}$ and the normalizing factor for Cauchy distribution. Thus
\begin{align*}
\sup_{\norm{u}=1} \norm{\nabla_{u} f(z)} & = \sup_{\norm{u}=1} \norm{\sum_{i=1}^{n} \left(\frac{\partial}{\partial z_{i}} f(z)\right) u_{i}} \\
& = \sup_{\norm{u}=1} \norm{\sum_{i=1}^{n} c \cdot \frac{-(n+1)}{2} \cdot \left(1 + \sqnorm{z}\right)^{-(n+3)/2} \cdot 2 z_{i} u_{i}} \\
& = \frac{c(n+1) \norm{z}}{\left(1+\sqnorm{z}\right)^{(n+3)/2}}, & \text{using $u = \frac{-z}{\norm{z}}$}.
\end{align*}
Plugging this in Equation \eqref{eq:log-ratio} we get
\begin{align*}
\log \frac{f(x)}{f(y)} & \leq \norm{x-y} \sup_{z \in \reals^{n}} \frac{(n+1) \norm{z}}{1+\sqnorm{z}} \\
& = \norm{x-y} (n+1) \sup_{z \in \reals^{n}} \frac{\norm{z}}{1+\sqnorm{z}} \\
& = \norm{x-y} (n+1) \sup_{r \in \reals} \frac{r}{1+r^{2}} \\
& = \frac{\norm{x-y} (n+1)}{2} \\
& \leq 1 & \text{using $\norm{x-y} \leq 1/n$}.
\end{align*}
Therefore
\[
\frac{f(y)}{f(x)} \geq \frac{1}{e}.
\]
\end{proof}
\end{document}